\documentclass{birkjour}

\usepackage[T2A]{fontenc}
\usepackage[english]{babel}
\usepackage{latexsym}
\usepackage{amssymb}
\usepackage{amsmath}

 \newtheorem{thm}{Theorem}[section]

 \theoremstyle{definition}
 
 \theoremstyle{remark}

 \numberwithin{equation}{section}

\def\I{{\rm I}}
\def\cl{{C}\!\ell}
\def\R{{\Bbb R}}

\def\Adj{{\rm Adj}}
\def\First{{\rm First}}
\def\Last{{\rm Last}}
\def\Even{{\rm Even}}
\def\Odd{{\rm Odd}}

\begin{document}

%
%
%
%
%
%
%
%
%

\title[Method of averaging in Clifford algebras]
 {Method of Averaging\\ in Clifford Algebras}

\author[D.~S.~Shirokov]{D.~S.~Shirokov}

\address{
1. A. A. Kharkevich Institute for Information Transmission Problems, Russian Academy of Sciences,\\
Bolshoy Karetny per. 19, 127994, Moscow, Russia.\\
2. N. E. Bauman Moscow State Technical University,\\
ul. Baumanskay 2-ya, 5, 105005, Moscow, Russia.
}

\email{dm.shirokov@gmail.com}

\thanks{This work was supported by Russian Science Foundation (project RSF 14-11-00687, Steklov Mathematical Institute).
}


\date{March 29, 2015}

\begin{abstract}
In this paper we consider different operators acting on Clifford algebras. We consider Reynolds operator of Salingaros' vee group. This operator ``average'' an action of Salingaros' vee group on Clifford algebra. We consider conjugate action on Clifford algebra.

We present a relation between these operators and projection operators onto fixed subspaces of Clifford algebras.
Using method of averaging we present solutions of system of commutator equations.
\end{abstract}

\maketitle

\section{Introduction}

In this paper we consider the following operators acting on real Clifford algebras
\begin{eqnarray}
F_S(U)=\frac{1}{|S|}\sum_{A\in S}(e^A)^{-1}U e^A,\label{Reyn0}
\end{eqnarray}
where
$$e^A= e^{a_1 a_2\ldots a_k} = e^{a_1} e^{a_2} \cdots  e^{a_k},\qquad A = a_1 a_2 \ldots a_k,\qquad a_1 < a_2 < \cdots < a_k,$$
are basis elements generated by an orthonormal basis in vector space $V$ over a field $\R$.
Here $S\subseteq \I$ is a subset of the set of all ordered multi-indices $A$ of the length from $0$ to $n$. We denote the number of elements in $S$ by $|S|$. Note that not for every subset $S\subseteq I$ in (\ref{Reyn0}), the set $\{e^A | A\in S\}$ is a group.

We can consider Reynolds operator (see, for example, \cite{CLS}) acting on a Clifford algebra element $U\in\cl(p,q)$
\begin{eqnarray}
R_G(U)=\frac{1}{|G|}\sum_{g\in G}g^{-1} U g,\label{Reyn}
\end{eqnarray}
where $|G|$ is the number of elements in a finite subgroup $G\subset \cl(p,q)^\times$. We denote the group of all invertible Clifford algebra elements by $\cl(p,q)^\times$.
These operators ``average'' an action of group $G$ on Clifford algebra $\cl(p,q)$.

We can take Salingaros' vee group $G=\{\pm e^A, A\in\I\}$ (see \cite{Abl1}, \cite{Abl2}, \cite{Abl3}), where $e^A$ are basis elements of Clifford algebra $\cl(p,q)$. Note that Salingaros' vee group is a finite subroup of spin groups (groups ${\rm Pin(p,q)}$ and $\rm Spin(p,q)$, see \cite{Abl2}, \cite{Helm1}, \cite{Pauli3}, \cite{Pauli4}).

We can write in this case
$$\frac{1}{|G|}\sum_{g\in G}g^{-1} U g=\frac{1}{2^{n+1}}\sum_{A\in \I}((e^A)^{-1}Ue^A+ (-e^A)^{-1}U(-e^A))=\frac{1}{2^n}\sum_{A\in \I}(e^A)^{-1}Ue^A.$$
We consider such operators further in this paper.

Note that operators (\ref{Reyn}) are often used in representation theory of finite groups (see \cite{Serr}, \cite{Dixon}, \cite{Babai}).
We use these operators in Clifford algebras to obtain some new properties.

We present a relation between these operators and projection operators onto fixed subspaces of Clifford algebras.

Using method of averaging we present solutions $X$ of the system of commutator equations
$$e^A X+\epsilon X e^A=Q^A,\qquad A\in S\subseteq \I,\qquad \epsilon\in\R^\times$$
for some given elements\footnote{Note that $A$ in $Q^A$ is a label (multi-index) that serves to pair off an arbitrary Clifford element $Q^A$ with basis element $e^A$.} $Q^A\in\cl(p,q)$. We use notation $\R^\times=\R\setminus 0$.

\section{Clifford algebras, ranks, projection operators}

Consider real Clifford algebra $\cl(p,q)$ with $p+q=n$, $n\geq1$.
The construction of Clifford algebra is discussed in details in \cite{Lounesto}, \cite{Helm2}, \cite{Lam} or \cite{Marchuk:Shirokov}.

Let $e$ be the identity element and let $e^a$, $a=1,\ldots,n$ be generators\footnote{We use notation from \cite{Benn:Tucker} (see, also \cite{Marchuk}). Note that there exists another notation instead of $e^a$ - with lower indices. But we use upper indices because we take into account relation with differential forms. Note that $e^a$ is not exponent.} of the Clifford algebra $\cl(p,q)$,
$$
e^a e^b+e^b e^a=2\eta^{ab}e,
$$
where $\eta=||\eta^{ab}||=||\eta_{ab}||$ is the diagonal matrix with $p$ pieces of $+1$ and $q$ pieces of $-1$ on the diagonal. Elements
$$
e^{a_1\ldots a_k}=e^{a_1}\cdots e^{a_k},\qquad a_1<\cdots<a_k,\,k=1,\ldots,n,
$$
together with the identity element $e$ form the basis of the Clifford
algebra. The number of basis elements is equal to $2^n$.

Let us denote the set of ordered multi-indices of the length from $0$ to $n$ by
\begin{eqnarray}
\I=\{-,\, 1,\, \ldots,\, n,\, 12,\, 13,\, \ldots,\, 1\ldots n\},\label{Inab}
\end{eqnarray}
where ``--'' is an empty multi-index. So, we have the basis of Clifford algebra $\mathfrak{B}=\{e^A,\, A\in\I \}$, where $A$ is an arbitrary ordered multi-index\footnote{We use notation $e^A$ from \cite{Benn:Tucker}. Note that $e^A$ is not exponent, $A$ is a multi-index.}. Let us denote the length of multi-index $A$ by $|A|$. So, we use notation
$$e^A= e^{a_1 a_2\ldots a_k} = e^{a_1} e^{a_2} \cdots  e^{a_k},\qquad A = a_1 a_2 \ldots a_k,\qquad a_1 < a_2 < \cdots < a_k.$$

Below we also consider different subsets $S\subseteq \I$:
$$\I_{\Even}=\{A\in\I,\, |A| - \mbox{even}\},\qquad \I_{\Odd}=\{A\in\I,\, |A| - \mbox{odd}\}.$$

We have $e_a=\eta_{ab}e^{b}$, $e^a=\eta^{ab}e_b$, where we use Einstein summation convection (there is a sum over index $b$).
We denote\footnote{We want to deal only with ordered multi-indices. So, multi-index $a_1 \ldots a_k$ is indivisible object in our consideration. It is not a set of indices $a_1, a_2, \ldots a_k$. That's why $e_{a_1 \ldots a_k}\neq e_{a_1}\cdots e_{a_k}$ in our notation.} expressions
$$\eta_{a_1 b_1}\cdots \eta_{a_k b_k} e^{b_k}\cdots e^{b_1}=e_{a_k}\cdots e_{a_1}=(e^{a_1\ldots a_k})^{-1},\quad a_1<\cdots< a_k.$$
by $e_{a_1\ldots a_k}$ (not $e_{a_k\ldots a_1}$). So, $e_A=(e^A)^{-1}\,$ $\forall A\in\I$ in our notation.

Any Clifford algebra element\footnote{We denote Clifford algebra elements by capital letters (not small letters that is more traditional) to avoid confusion with numbers because sometimes Clifford algebra elements have indices or multi-indices in this paper too (see \cite{Marchuk}).} $U\in\cl(p,q)$ can be written in the form
\begin{eqnarray}
U=ue+u_a e^a+\sum_{a_1<a_2}u_{a_1 a_2}e^{a_1 a_2}+\cdots+u_{1\ldots n}e^{1\ldots n}=u_A e^A,\label{U}
\end{eqnarray}
where we have a sum\footnote{We use Einstein summation convection for multi-indices too.} over ordered multi-index $A$ and $\{u_A\}=\{u, u_a, u_{a_1 a_2}, \ldots, u_{1\ldots n}\}$ are real numbers.

We denote by $\cl_k(p,q)$, $k=0, 1, \ldots, n$ the vector spaces that span over the basis elements
$e^{a_1\ldots a_k}$. Elements of $\cl_k(p,q)$ are said to be
elements of rank\footnote{There is a difference in notation in literature. We use term ``rank'' and notation $\cl_k(p,q)$ because we take into account relation with differential forms, see \cite{Marchuk}.} $k$. We have
\begin{eqnarray}
\cl(p,q)=\bigoplus_{k=0}^{n}\cl_k(p,q).\label{ranks}
\end{eqnarray}
We consider (linear) projection operators on the vector subspaces $\cl_k(p,q)$
\begin{eqnarray}
\pi_k: \cl(p,q)\to \cl_k(p,q),\qquad \pi_k(U)=\sum_{a_1<\cdots<a_k}u_{a_1\ldots a_k}e^{a_1 \ldots a_k}.\label{proj}
\end{eqnarray}

Clifford algebra $\cl(p,q)$ is a superalgebra. It is represented as the direct sum of even and odd subspaces (of {\it even} and {\it odd} elements respectively)
\begin{eqnarray}
\cl(p,q)=\cl_{\Even}(p,q)\oplus\cl_{\Odd}(p,q),\nonumber
\end{eqnarray}
$$\cl_{\Even}(p,q)=\bigoplus_{k - even}\cl_k(p,q),\qquad \cl_{\Odd}(p,q)=\bigoplus_{k - odd}\cl_k(p,q).$$

\section{Reynolds operator of the Salingaros' vee group}

We have the following well-known statement about center
$${\rm Cen}(\cl(p,q))=\{U\in \cl(p,q) \,|\, UV=VU\quad \forall V\in\cl(p,q)\}$$
of Clifford algebra $\cl(p,q)$.

\begin{thm}\label{theoremCentr} We have
\begin{eqnarray}
{\rm Cen}(\cl(p,q))=\left\lbrace\begin{array}{ll}
\cl_0(p,q), & \mbox{$n$ is even};\\
\cl_0(p,q)\oplus \cl_n(p,q), & \mbox{$n$ is odd.}
\end{array}
\right.
\end{eqnarray}
\end{thm}

Let us consider the following operator (Reynolds operator of the Salingaros' vee group $\{e^A, A\in\I\}$, see above)
$$F(U)=\frac{1}{2^n} e_A U e^A,$$
where we have a sum over multi-index $A\in\I$.

\begin{thm}\label{theoremSvPoln} We have
\begin{eqnarray}F(U)=\frac{1}{2^n}e_A U e^A=\left\lbrace
\begin{array}{ll}
\pi_0(U), & \parbox{.2\linewidth}{if $n$ is even;} \\
\pi_0(U)+\pi_n(U), & \parbox{.2\linewidth}{if $n$ is odd,}
\end{array}
\right.\label{poln}
\end{eqnarray}
where $\pi_0$ and $\pi_n$ are projection operators (see (\ref{proj})) onto the subspaces of fixed ranks. Operator $F$ is a projector $F^2=F$ (on the center of Clifford algebra).
\end{thm}

\begin{proof} We have
$$(e^a)^{-1} F(U) e^a=\sum_{A} (e^{A} e^a)^{-1} F(U) (e^{A} e^a)=\sum_{B} (e^{B})^{-1} F(U) e^{B}=F(U).$$
So, $F(U)$ is in the center of Clifford algebra (see Theorem \ref{theoremCentr}).
For elements $U$ of ranks $k=1, \ldots, n-1$ (and $k=n$ in the case of even $n$) we have $F(U)=0$.
In other particular cases we have $e_A e^A=2^n e$ and (in the case of odd $n$) $e_A e^{1\ldots n} e^A=2^n e^{1\ldots n}$.
It is also easy to verify that $F^2=F$.
\end{proof}

Note that ${\rm Cen}(\cl(p,q))$ is the ``ring of invariants'' (in the language of \cite{CLS}) of Salingaros' vee group.

\begin{thm} Let an element $X\in \cl(p,q)$ satisfy the system of $2^n$ equations with some given elements\footnote{Note that $A$ in $Q^A$ is a label (multi-index) that serves to pair off an arbitrary Clifford element $Q^A$ with basis element $e^A$.} $Q^A\in\cl(p,q)$
\begin{eqnarray}
e^A X+\epsilon X e^A=Q^A \quad \forall A\in\I,\qquad \epsilon\in\R^\times.\label{se1}
\end{eqnarray}
If $\epsilon=-1$ (commutator case), then this system of equations either has no solution or it has a unique solution up to element of center:
\begin{eqnarray}
X=-\frac{1}{2^n}Q^A e_A+Z,\qquad Z\in{\rm Cen}(\cl(p,q)).
\end{eqnarray}

If $\epsilon\neq -1$, then this system of equations either has no solution or it has a unique solution
\begin{eqnarray}
X=
\left\lbrace
\begin{array}{ll}
\frac{1}{2^{n}\epsilon}\left(Q^Ae_A-\frac{1}{(\epsilon+1)}\pi_0(Q^A e_A)\right), & \mbox{if $n$ is even},\\\\
\frac{1}{2^{n}\epsilon}\left(Q^Ae_A-\frac{1}{(\epsilon+1)}(\pi_0(Q^A e_A)+\pi_n(Q^A e_A))\right),  & \mbox{if $n$ is odd}.
\end{array}
\right.
\end{eqnarray}
\end{thm}

\begin{proof} Let us multiply each equation by $e_A$ on the right and add them (see Theorem \ref{theoremSvPoln}):
$$e^A X e_A+\epsilon X e^A e_A=Q^A e_A \quad\Rightarrow\quad 2^n \pi_{center}(X)+\epsilon X 2^n =Q^A e_A,$$
where $\pi_{center}$ is the projection on the center of Clifford algebra. Using $X=\sum_{k=0}^n\pi_k(X)$ and Theorem \ref{theoremCentr}, we obtain statement of the theorem.
\end{proof}

Note that we have a solution or have no solution of the system of commutator equations. It depends on elements $Q^A$ (it suffices to substitute solution in equation and check the equality).

As suggested by one of the referees, in the case $\epsilon\neq -1$ we can take solution of the ``first'' equation (\ref{se1}) $X=\frac{1}{1+\epsilon}Q$ (here $Q$ has empty multi-index) and substitute it in the other equations. We obtain the condition
\begin{eqnarray}Q^A=\frac{1}{1+\epsilon}(e^A Q+\epsilon Q e^A),\label{qa}\end{eqnarray}
that we can also rewrite in the form\footnote{We use operators $\pi_{[A]}$ and $\pi_{\{A\}}$ here. See about them below.} $(e^A)^{-1}Q^A=\pi_{[A]}(Q)+\frac{1-\epsilon}{1+\epsilon}\pi_{\{A\}}(Q)$.

So, we can say, that system of equations (\ref{se1}) has solution in the case $\epsilon\neq -1$ if and only if given elements $Q^A$ have the following connection (\ref{qa}) with the first of them ($Q$).

In the case $\epsilon=-1$ it is not difficult to understand that $\pi_{center}(Q^A (e^A)^{-1})=0$ is the necessary condition for the system (\ref{se1}) to have a solution. Sufficient condition is
$$-e^A (\frac{1}{2^n}Q^B e_B)+(\frac{1}{2^n}Q^B e_B) e^A=Q^A,\qquad \forall A\in\I$$
with summation over multi-index $B$. It can be rewritten in the form $$(e^A)^{-1}Q^A=\frac{1}{2^n}\pi_{\{A\}}(e_B Q^B),\qquad \forall A\in\I.$$

\section{Adjoint sets of multi-indices}

We call ordered multi-indices $a_1 \ldots a_k$ and $b_1 \ldots b_l$ {\it adjoint multi-indices} if they have no common indices and they form multi-index $1\ldots n$ of the length $n$. We write $b_1\ldots b_l=\widetilde{a_1 \ldots a_k}$ and $a_1 \ldots a_k=\widetilde{b_1 \ldots b_l}$. We call corresponding basis elements
$e^{a_1 \ldots a_m}$, $e^{b_1 \ldots b_l}$ {\it adjoint} and write $e^{b_1 \ldots b_l}=e^{\widetilde{a_1 \ldots a_m}}$, $e^{a_1 \ldots a_m}=e^{\widetilde{b_1 \ldots b_l}}$. We can also write that $e^{a_1 \ldots a_m} e^{b_1 \ldots b_l}=\pm e^{1\ldots n}$ and
$\star e^{a_1 \ldots a_m} = \pm e^{b_1 \ldots b_l}$, where $\star$ is Hodge operator\footnote{It is the analogue of Hodge operator in Clifford algebra $\star U=\tilde{U} e^{1\ldots n}$, where $\tilde{}$ is the reversion anti-automorphism in the Clifford algebra $\cl(p,q)$ \cite{Lounesto}.}.
We denote the sets of corresponding $2^{n-1}$ multi-indices by $\I_{\Adj}$ and $\widetilde{\I_{\Adj}}=\I \setminus \I_{\Adj}$. So, for each multi-index in $\I_{\Adj}$ there exists adjoint multi-index in $\widetilde{\I_{\Adj}}$. We have
\begin{eqnarray}
\mathfrak{B}=\{e^A \,|\, A\in\I \}=\{e^A \,|\, A\in\I_{\Adj}\} \cup \{e^A \,|\, A\in \widetilde{\I_{\Adj}}\}.\label{razb}
\end{eqnarray}

For Clifford algebra $\cl(p,q)$ of dimension $n=p+q$ we have $2^{2^{n-1}-1}$ different partitions\footnote{For example, in the case $n=2$ we have 2 partitions $\{e, e^1, e^2, e^{12}\}=\{e, e^1\} \cup \{e^{12}, e^{2}\}=\{e, e^2\}\cup \{e^{12}, e^2\}$.} of the form (\ref{razb}). For example,
$$\I_{\Adj}=\I_{\First},\qquad \widetilde{\I_{\Adj}}=\I\setminus\I_{\First}=\I_{\Last},$$
where $\I_{\First}$ consists of the first (in the order) $2^{n-1}$ multi-indices of the set $\I$ (\ref{Inab}).
In the case of odd $n$ we can write
$$\I_{\First}=\{A\in\I,\quad |A|\leq\frac{n-1}{2}\},\qquad \I_{\Last}=\{A\in\I,\quad |A|\geq\frac{n+1}{2}\}.$$

In the case of odd $n$ we can consider the following adjoint sets
$$\I_{\Adj}=\I_{\Even},\qquad \widetilde{\I_{\Adj}}=\I_{\Odd}.$$

\section{Commutative properties of basis elements}

\begin{thm}\label{theoremCommBas2} Consider real Clifford algebra $\cl(p,q)$, $p+q=n$ and the set of basis elements $\mathfrak{B}=\{e^A, A\in\I\}$.

Then each element of this set (if it is neither $e$ nor $e^{1\ldots n}$) commutes with $2^{n-2}$ even elements of the set $\mathfrak{B}$, commutes with $2^{n-2}$ odd elements of the set $\mathfrak{B}$, anticommutes with $2^{n-2}$ even elements of the set $\mathfrak{B}$ and anticommutes with $2^{n-2}$ elements of the set $\mathfrak{B}$.
Element $e$ commutes with all elements of the set $\mathfrak{B}$.
\begin{enumerate}
  \item if $n$ - even, then $e^{1\ldots n}$ commutes with all $2^{n-1}$ even elements of the set $\mathfrak{B}$ and anticommutes with all $2^{n-1}$ odd elements of the set $\mathfrak{B}$;
  \item if $n$ - odd, then $e^{1\ldots n}$ commutes with all $2^n$ elements of the set $\mathfrak{B}$.
\end{enumerate}
\end{thm}

\begin{proof} The cases $k=0$ and $k=n$ are trivial (see Theorem \ref{theoremCentr}).

Let us fix one multi-index $A$ of the length $k$. Then there are $C^i_k C_{n-k}^{m-i}$ different multi-indices of the fix length $m$ that have fixed number $i$ coincident indices with multi-index $A$. Here $C_n^k=\left(\begin{array}{l} n  \\ k  \end{array}\right)=\frac{n!}{k!(n-k)!}$ is binomial coefficient (we have $C_n^k=0$ for $k>n$). Note, that $\sum_{i=0}^nC^i_k C_{n-k}^{m-i}=C_n^m$ (Vandermonde's convolution) - the full number of ordered multi-indices of the length $m$. When we swap basis element with multi-index $A$ of the length $k$ with another basis element with multi-index of the length $m$, then we obtain coefficient $(-1)^{km-i}$, where $i$ is the number of coincident indices in these 2 multi-indices, i.e. $e^{a_1 \ldots a_k}e^{b_1 \ldots b_m}=(-1)^{km-i}e^{b_1 \ldots b_m}e^{a_1 \ldots a_k}$.

If $k$ is even and does not equal to $0$ and $n$, then the number of even and odd elements $e^{b_1 \ldots b_m}$ that commute (in this case coefficient $km-i$ must be even, and so $i$ is even) with fixed $e^{a_1 \ldots a_k}$ respectively equals
$$\sum_{m-even}\sum_{i - even}C^i_k C_{n-k}^{m-i}=2^{n-2},\qquad \sum_{m-odd}\sum_{i - even}C^i_k C_{n-k}^{m-i}=2^{n-2}.$$
If $k$ is odd and does not equal to $n$, then the number of even and odd elements $e^{b_1 \ldots b_m}$ that anticommute (in this case $km-i$ is odd, and so $m-i$ is even) with $e^{a_1 \ldots a_k}$ respectively equals
$$\sum_{m - even} \sum_{i - even}C^i_k C_{n-k}^{m-i}=2^{n-2},\qquad \sum_{m - odd} \sum_{i - odd}C^i_k C_{n-k}^{m-i}=2^{n-2}.$$

We can prove the last 4 identities if we regroup summands. For example, we have
$$\sum_{m-even}\sum_{i - even}C^i_k C_{n-k}^{m-i}=(\sum_{j= even}C_k^j)(\sum_{l-even}C_{n-k}^l)=2^{k-1}2^{n-k-1}=2^{n-2}.$$
\end{proof}

Also we have the following theorem about adjoint sets of multi-indices.

\begin{thm}\label{theoremCommBas3} Consider real Clifford algebra $\cl(p,q)$, $p+q=n$ and the set of basis elements $\mathfrak{B}=\{e^A, A\in\I\}$. Suppose we have a partition $\I=\I_{\Adj} \cup \widetilde{\I_{\Adj}}$.

If $n$ is even then any even (not odd!) basis element (if it is not $e$) commutes with $2^{n-2}$ basis elements from $\{e^A \,|\, A\in\I_{\Adj}\}$, anticommutes  with $2^{n-2}$ basis elements from $\{e^A \,|\, A\in\I_{\Adj}\}$, commutes with $2^{n-2}$ basis elements from $\{e^A \,|\, A\in \widetilde{\I_{\Adj}}\}$ and anticommutes with $2^{n-2}$ basis elements from $\{e^A \,|\, A\in \widetilde{\I_{\Adj}}\}$.

If $n$ is odd then any basis element (if it is neither $e$ nor $e^{1\ldots n}$) commutes with $2^{n-2}$ basis elements from $\{e^A \,|\, A\in\I_{\Adj}\}$, anticommutes  with $2^{n-2}$ basis elements from $\{e^A \,|\, A\in\I_{\Adj}\}$, commutes with $2^{n-2}$ basis elements from $\{e^A \,|\, A\in \widetilde{\I_{\Adj}}\}$ and anticommutes with $2^{n-2}$ basis elements from $\{e^A \,|\, A\in \widetilde{\I_{\Adj}}\}$.
\end{thm}
Note that in the case of odd $n$ we can take $\I_{\Adj}=\I_{\Even}$, $\widetilde{\I_{\Adj}}=\I_{\Odd}$ and obtain the statement from the Theorem \ref{theoremCommBas2}.

\begin{proof} If $n$ is odd then $e^{1\ldots n}$ is in the center of Clifford algebra. So if basis element commutes with some basis element, then it commutes with adjoint basis element. But we know from Theorem \ref{theoremCommBas2} that basis elements (except $e$ and $e^{1\ldots n}$) commutes with $2^{n-1}$ basis elements and anticommutes with $2^{n-1}$ basis elements. So we obtain the statement of theorem for the case of odd $n$.

If $n$ is even then even (not odd) basis element commutes with $e^{1\ldots n}$. So if even basis element commutes with some basis element, then it commutes with adjoint basis element.
\end{proof}

Let us represent the commutative property of basis elements in the following tables. At the intersection of two basis elements is a sign ``+'' if they commute and the sign ``--'' if they anticommute. For small dimensions we have the following tables:

\begin{center}
\begin{tabular}{|c||c|c|}
\hline{{\vrule width0pt height15pt}}
$n=1$  & $e$ & $e^1$ \\ \hline \hline
$e$ & + & +  \\ \hline
$e^1$ & + & +  \\ \hline
\end{tabular}\qquad
\qquad
\begin{tabular}{|c||c|cc|c|}
\hline{{\vrule width0pt height15pt}}
$n=2$  & $e$ & $e^1$ & $e^2$ & $e^{12}$ \\ \hline \hline
$e$ & + &  +  & + & +\\ \hline
$e^1$ & + & + & --& -- \\
$e^2$ & + & --&  +& -- \\ \hline
$e^{12}$ & +  & -- & -- & + \\ \hline
\end{tabular}
\end{center}

\begin{center}
\begin{tabular}{|c||c|ccc|ccc|c|}
\hline{{\vrule width0pt height15pt }}
$n=3$  & $e$ & $e^1$ & $e^2$ & $e^{3}$ & $e^{12}$ & $e^{13}$ & $e^{23}$ & $e^{123}$ \\ \hline \hline
$e$      & + & + &  + & + & + & + & + & + \\ \hline
$e^1$    & + & + &  -- & -- & -- & -- & + & + \\
$e^2$    & + & --  & + & -- & --  &+ & -- & + \\
$e^3$    & + & --  & -- & + & + & -- & -- & + \\ \hline
$e^{12}$ & + & -- & -- & + & + & -- & -- & + \\
$e^{13}$ & + & -- &  + & -- & -- & + & -- & + \\
$e^{23}$ & + & +  & -- & -- & -- & -- & + & + \\ \hline
$e^{123}$& + & +  & + & + & + & + & + & + \\ \hline
\end{tabular}
\end{center}


Consider the following operator
$$F_{\Adj}(U)=\frac{1}{2^{n-1}}\sum_{A\in\I_{\Adj}}e_A U e^A.$$

\begin{thm}\label{theoremSvDual} Consider an arbitrary Clifford algebra element $U$. Suppose we have a partition $\I=\I_{\Adj} \cup \widetilde{\I_{\Adj}}$.
In the case of arbitrary $n$ we have
$$F_{\Adj}(U)=F(U).$$
\end{thm}

\begin{proof} If $n$ is odd, then $e^{1\ldots n}$ is in the center of Clifford algebra,\\
$(e^{a_1\ldots a_m})^{-1}Ue^{a_1 \ldots a_m}=e^{1\ldots n}(e^{1\ldots n})^{-1}(e^{a_1\ldots a_m})^{-1}Ue^{a_1 \ldots a_m}=(e^{\widetilde{a_1\ldots a_m}})^{-1}Ue^{\widetilde{a_1\ldots a_m}},$
and
\begin{eqnarray}
e_A U e^A = 2\sum_{A\in\I_{\Adj}}e_A U e^A.\label{ert}
\end{eqnarray}

If $n$ is even, then $e^{1\ldots n}$ anticommutes with all odd basis elements and commutes with all even basis elements (see Theorem \ref{theoremCommBas2}). So
if $U=U_0+U_1$, $U_0\in\cl_{\Even}(p,q)$, $U_1\in\cl_{\Odd}(p,q)$, then for $k=0, 1$ we have
$$(e^{a_1\ldots a_m})^{-1}U_k e^{a_1 \ldots a_m}=e^{1\ldots n}(e^{1\ldots n})^{-1}(e^{a_1\ldots a_m})^{-1}U_k e^{a_1 \ldots a_m}=$$
$$=(-1)^{2m+k}(e^{a_1\ldots a_m})^{-1}(e^{1\ldots n})^{-1}U_k e^{1\ldots n}e^{a_1 \ldots a_m}=(-1)^k (e^{\widetilde{a_1\ldots a_m}})^{-1}U_k e^{\widetilde{a_1\ldots a_m}},$$
and we obtain (\ref{ert}) again.
\end{proof}

So we can use operator $F_{\Adj}$ (with $2^{n-1}$ summands) instead of operator $F(U)$ (with $2^{n}$ summands) in all calculations.

\section{Conjugate action on Clifford algebras}

We denote the corresponding square symmetric matrices of size $2^n$ from the previous section (with elements $1$ and $-1$, see tables) by $M_n=||m_{AB}||$. For arbitrary element of these matrices we have\footnote{Note that $e^A e^B(e^A)^{-1} (e^B)^{-1}$ is the group commutator of $e^A$ and $e^B$ in Salingaros' vee group.}  $m_{AB}=e^A e^B(e^A)^{-1} (e^B)^{-1}$, $A, B\in\I$, $e\equiv1$.
We have\footnote{We can say that commutator subgroup of Salingaros' vee group is $\{1, -1\}$.}
\begin{eqnarray}
m_{AB}=m_{BA}=\left\lbrace\begin{array}{lll}
1, & \mbox{if $[e^A, e^B]=0$};\\
-1, & \mbox{if $\{e^A, e^B\}=0$}.
\end{array}
\right.
\end{eqnarray}

In the case of odd $n$ we also consider symmetric matrix $L$ of size $2^{n-1}$ $L_n=||l_{AB}||$, $l_{AB}=m_{AB}$, $A, B\in\I_{\First}=\{A\in\I,\quad |A|\leq\frac{n-1}{2}\}.$

\begin{thm}\label{obrat} Matrix $M_n$ is invertible in the case of even $n$ and $M_n^{-1}=\frac{1}{2^n} M_n$.
Matrix $M_n$ is not invertible in the case of odd $n$.

Matrix $L_n$ is invertible in the case of odd $n$ and $L_n^{-1}=\frac{1}{2^{n-1}}L_n$.
\end{thm}

\begin{proof} Matrices are symmetric $M_n^T=M_n$, $L_n^T=L_n$ by definition. Let us multiply matrix $M_n$ by itself. For two arbitrary rows we have
$$\sum_B m_{AB}m_{BC}=\sum_B e^A e^B(e^A)^{-1} (e^B)^{-1}e^B e^C(e^B)^{-1} (e^C)^{-1}=$$
$$=e^A \left(\sum_Be^B(e^A)^{-1} e^C(e^B)^{-1}\right)(e^C)^{-1}.$$
In the last expression sum is equal to zero if $A\neq C$ and (in the case of odd $n$) $A, C$ are not adjoint multi-indices, because $e_B U e^B$ is projection onto the center of Clifford algebra (see Theorem \ref{theoremSvPoln}). In other cases the last expression equals to $2^n$. In the case of odd $n$ we must use matrix $L_n$ because we do not have adjoint multi-indices in this matrix.
\end{proof}

Let us consider the following operators (for different $A\in\I$)
$$F_{e^A}: \cl(p,q)\to\cl(p,q),\qquad U\to (e^A)^{-1}Ue^A,\qquad U\in\cl(p,q).$$
Note that $F_{e^A}(U)=(e^A)^{-1} U e^A$ is a {\it conjugation} of Clifford algebra element $U\in\cl(p,q)$ by element $e^A$ of Salingaros' vee group.

\begin{thm}\label{theoremSimp1} For operator $F_{e^A}(U)=(e^A)^{-1} U e^A$ we have
\begin{eqnarray}
F_{e^A}(U)=\sum_{B}m_{AB} \pi_{e^B}(U),\label{odsv}
\end{eqnarray}
where $\pi_{e^B}$ is a projection\footnote{We do not use notation $\pi_B$ instead of $\pi_{e^B}$ because it will conflict with notation $\pi_k$ for projection onto subspace $\cl_k(p,q)$.} onto subspace spanned over element $e^B$.
We have $F_{e^A}(F_{e^A}(U))=U$.
\end{thm}

\begin{proof} The statement follows from the definition of matrix $M_n=||m_{AB}||$ and definition of conjugation. \end{proof}

Fixed multi-index $A$ divides the set $\I$ into 2 sets $\I=\I_{[A]}\cup \I_{\{A\}},$ where $e^B$, $B\in\I_{[A]}$ commute with $e^A$, and $e^B$, $B\in\I_{\{A\}}$ anticommute with $e^A$. Denote the corresponding subspaces of Clifford algebra by $\cl_{[A]}(p,q)$ and $\cl_{\{A\}}(p,q)$ and corresponding projection operators by $\pi_{[A]}$ and $\pi_{\{A\}}$. We have $\cl(p,q)=\cl_{[A]}(p,q)\oplus\cl_{\{A\}}(p,q)$ and
$$F_{e^A}(U)=(e^A)^{-1} U e^A=\pi_{[A]}(U)-\pi_{\{A\}}(U), \qquad \forall A.$$

\begin{thm} For arbitrary Clifford algebra element $U$ we have
$$\pi_{[A]}(U)=\frac{1}{2}(U+(e^A)^{-1 }U e^A),\qquad \pi_{\{A\}}(U)=\frac{1}{2}(U-(e^A)^{-1 }U e^A).$$
\end{thm}

\begin{proof} Using
$$(e^A)^{-1 }U e^A=\pi_{[A]}(U)-\pi_{\{A\}}(U),\qquad U=\pi_{[A]}(U)+\pi_{\{A\}}(U)$$
we obtain the statement of theorem.
\end{proof}

For empty multi-index $A=-$ we have $m_{-,B}=1$ for all $B$, $\I=\I_{[A]}$, $\I_{\{A\}}=\o$.
For multi-index $A=1 \ldots n$ we have $m_{1\ldots n, B}=1$ for all $B$ in the case of odd $n$ and
\begin{eqnarray}
m_{1\ldots n, B}=\left\lbrace\begin{array}{lll}
1, & \mbox{if $B$ is even};\\
-1, & \mbox{if $B$ is odd},
\end{array}
\right.
\end{eqnarray}
and $e_{1\ldots n}Ue^{1\ldots n}=\pi_{\Even}(U)-\pi_{\Odd}(U)$ in the case of even $n$, where $\pi_{\Even}$ and $\pi_{\Odd}$ are projection operations onto the even and odd subspaces of Clifford algebra. In other cases (when $A$ is not empty and in not $1\ldots n$) we have $2^{n-1}$ elements in each of the sets
$\I_{[A]}$, $\I_{\{A\}}$ (see Theorem \ref{theoremCommBas2}) i.e. we have
$\dim\cl_{[A]}(p,q)=\dim\cl_{\{A\}}(p,q)=2^{n-1}$ in these cases.

In particular case we obtain the following identities (for $A=1\ldots n$):
in the case of even $n$ we have
$$\pi_{\Even}(U)=\frac{1}{2}(U+e_{1\ldots n}Ue^{1\ldots n}),\qquad \pi_{\Odd}(U)=\frac{1}{2}(U-e_{1\ldots n}Ue^{1\ldots n}).$$
We have the following theorem.

\begin{thm}
Let an element $X\in \cl(p,q)$ satisfy the following equation with some given element\footnote{Note that $A$ in $Q^A$ is a label (multi-index) that serves to pair off an arbitrary Clifford element $Q^A$ with basis element $e^A$.} $Q^A\in\cl(p,q)$
\begin{eqnarray}
e^A X +\epsilon X e^A=Q^A,\qquad \epsilon\in\R^\times.
\end{eqnarray}
If $\epsilon\neq \pm 1$, then we have a unique solution
$$X=\sum_B \frac{1}{1+\epsilon m_{AB}}\pi_{e^B}((e^A)^{-1} Q^A).$$
If $\epsilon=-1$ (commutator case), then:
\begin{itemize}
  \item if $\pi_{[A]}((e^A)^{-1}Q^A)\neq 0$ (i.e. $\{e^A, Q^A\}\neq 0$), then there is no solution;
  \item if $\pi_{[A]}((e^A)^{-1}Q^A)=0$ (i.e. $\{e^A, Q^A\}=0$), then the solution is
  $$\frac{1}{2} \pi_{\{A\}}((e^A)^{-1} Q^A)+\pi_{[A]}(U),$$
  where $U$ is an arbitrary Clifford algebra element.
\end{itemize}
If $\epsilon=1$ (anticommutator case), then:
\begin{itemize}
  \item if $\pi_{\{A\}}((e^A)^{-1}Q^A)\neq 0$ (i.e. $[e^A, Q^A]\neq 0$), then there is no solution;
  \item if $\pi_{\{A\}}((e^A)^{-1}Q^A)=0$ (i.e. $[e^A, Q^A]= 0$), then the solution is
  $$\frac{1}{2} \pi_{[A]}((e^A)^{-1} Q^A)+\pi_{\{A\}}(U),$$
  where $U$ is an arbitrary Clifford algebra element.
\end{itemize}
\end{thm}

\begin{proof}
Multiply equation on the left by $(e^A)^{-1}=e_A$ and use Theorem \ref{theoremSimp1}:
$$X+\epsilon (e^A)^{-1} X e^A=(e^A)^{-1} Q^A \quad\Rightarrow\quad X+\epsilon\sum_B m_{AB} \pi_{e^B}(X) =(e^A)^{-1} Q^A.$$

Using $X=\sum_{B} \pi_{e^B}(X)$ we obtain
$$\sum_B (1+\epsilon m_{AB})\pi_{e^B}(X)=\sum_B\pi_{e^B}((e^A)^{-1}Q^A).$$
In the case $\epsilon\neq \pm 1$ we obtain the statement of the theorem.

Let $\epsilon=-1$. Then
$$X-(e^A)^{-1} X e^A=(e^A)^{-1} Q^A \Rightarrow 2\pi_{\{A\}}(X)=\pi_{\{A\}}((e^A)^{-1} Q^A)+\pi_{[A]}((e^A)^{-1} Q^A)$$
and we obtain the statement of the theorem for this case. Note that
$$\pi_{[A]}((e^A)^{-1}Q^A)=\frac{1}{2}((e^A)^{-1}Q^A+(e^A)^{-1}(e^A)^{-1}Q^Ae^A)=$$
$$=\frac{1}{2}((e^A)^{-1}Q^A+e^A(e^A)^{-1}Q^A(e^A)^{-1})=\frac{1}{2}\{(e^A)^{-1}, Q^A\}=\pm\frac{1}{2}\{e^A, Q^A\},$$
because $(e^A)^{-1}=\pm e^A$.

In the case $\epsilon=1$ proof is similar. Analogously we have $$\pi_{\{A\}}((e^A)^{-1}Q^A)=\pm\frac{1}{2}[e^A, Q^A].$$
\end{proof}

\begin{thm}\label{theoremSimp} In the case of even $n$ we have
\begin{eqnarray}
\pi_{e^A}(U)=\frac{1}{2^n}\sum_B m_{AB}(e^{B})^{-1} U e^B.
\end{eqnarray}
In the case of odd $n$ we have
\begin{eqnarray}
\pi_{e^A, \widetilde{e^A}}(U)=\pi_{e^A}(U)+\pi_{\widetilde{e^A}}(U)=\frac{1}{2^{n-1}}\sum_{B\in\I_{\First}}l_{AB}(e^{B})^{-1} U e^B.
\end{eqnarray}
Note that we can use instead of $\I_{\First}$ any adjoint set $\I_{\Adj}$.
\end{thm}

\begin{proof} From (\ref{odsv}) we obtain
$$
\left( \begin{array}{l}
 F_{e}(U) \\
 F_{e^a}(U) \\
 \ldots \\
 F_{e^{1\ldots n}}(U) \end{array}\right)=
 M_n
 \left( \begin{array}{l}
 \pi_e(U) \\
 \pi_{e^1}(U) \\
 \ldots \\
 \pi_{e^{1\ldots n}}(U) \end{array}\right).
$$
Using Theorem \ref{obrat} in the case of even $n$ we obtain
$$
\left( \begin{array}{l}
 \pi_e(U) \\
 \pi_{e^1}(U) \\
 \ldots \\
 \pi_{e^{1\ldots n}}(U) \end{array}\right)=
 \frac{1}{2^{n}}M_n
 \left( \begin{array}{l}
 F_{e}(U) \\
 F_{e^1}(U) \\
 \ldots \\
 F_{e^{1\ldots n}}(U) \end{array}\right).
$$
In the case of odd $n$ we have
\begin{eqnarray}
F_{e^A}(U)=(e^A)^{-1} U e^A=\sum_{B}m_{AB} \pi_{e^B}(U)=\sum_{B\in\I_{\First}}l_{AB}\pi_{e^B, \widetilde{e^{B}}}(U).
\end{eqnarray}
and use Theorem \ref{obrat}.
\end{proof}

Let us give some examples.

In the case $n=1$ we have
\begin{eqnarray}
M_1=\begin{pmatrix}1 & 1 \cr 1 &1\end{pmatrix},\qquad L_1=\begin{pmatrix} 1 \end{pmatrix},\qquad F_{e}(U)=F_{e^1}(U)=U,\qquad \pi_{e, e^1}(U)=U.\nonumber
\end{eqnarray}
In the case $n=2$ we have
$
M_2=\begin{pmatrix}1 & 1 &1 &1 \cr 1 &1&-1&-1\cr1&-1&1&-1\cr1&-1&-1&1\end{pmatrix},\nonumber
$
\begin{eqnarray}
F_{e}(U)=U,\qquad F_{e^1}(U)=\pi_e(U)+\pi_{e^1}(U)-\pi_{e^2}(U)-\pi_{e^{12}}(U),\nonumber\\
F_{e^2}(U)=\pi_e(U)-\pi_{e^1}(U)+\pi_{e^2}(U)-\pi_{e^{12}}(U),\nonumber\\
F_{e^{12}}(U)=\pi_e(U)-\pi_{e^1}(U)-\pi_{e^2}(U)+\pi_{e^{12}}(U).\nonumber
\end{eqnarray}
\begin{eqnarray}
\pi_{e}(U)&=&\frac{1}{4}(e_A Ue^A),\nonumber\\
\pi_{e^1}(U)&=&\frac{1}{4}(eUe+(e^1)^{-1} U e^1-(e^2)^{-1} U e^2-(e^{12})^{-1} U e^{12}),\nonumber\\
\pi_{e^2}(U)&=&\frac{1}{4}(eUe-(e^1)^{-1} U e^1+(e^2)^{-1} U e^2-(e^{12})^{-1} U e^{12}),\nonumber\\
\pi_{e^{12}}(U)&=&\frac{1}{4}(eUe-(e^1)^{-1} U e^1-(e^2)^{-1} U e^2+(e^{12})^{-1} U e^{12}).\nonumber
\end{eqnarray}
In the case $n=3$ we have
$
L_2=\begin{pmatrix} 1 &1&1&1\cr1 &1&-1&-1\cr1 &-1&1&-1\cr1 &-1&-1&1\cr\end{pmatrix},\nonumber
$
\begin{eqnarray}
\pi_{e, e^{123}}(U)&=&\frac{1}{4}(eUe+(e^1)^{-1} U e^1+(e^2)^{-1} U e^2+(e^{3})^{-1} U e^{12}),\nonumber\\
\pi_{e^1, e^{23}}(U)&=&\frac{1}{4}(eUe+(e^1)^{-1} U e^1-(e^2)^{-1} U e^2-(e^{3})^{-1} U e^{3}),\nonumber\\
\pi_{e^2,e^{13}}(U)&=&\frac{1}{4}(eUe-(e^1)^{-1} U e^1+(e^2)^{-1} U e^2-(e^{3})^{-1} U e^{3}),\nonumber\\
\pi_{e^3, e^{12}}(U)&=&\frac{1}{4}(eUe-(e^1)^{-1} U e^1-(e^2)^{-1} U e^2+(e^{3})^{-1} U e^{3}).\nonumber
\end{eqnarray}

\section{Conclusion}

In the present paper we consider different operators
$$F_S(U)=\frac{1}{|S|}\sum_{A\in S}(e^A)^{-1}U e^A,$$
acting on real Clifford algebra $\cl(p,q)$. Note that all theorems of this paper can be reformulated without changes for complex Clifford algebra.

We can also consider operators with other subsets $S\subseteq \I$. Note that not for every subset $S\subseteq \I$, the set $\{e^A | A\in S\}$ is a group.

In \cite{Marchuk:Shirokov} we consider the following operator in Clifford algebra $\cl(p,q)$
$$F_1(U)=e_a U e^a$$
and prove that $e_a U e^a=\sum_{k=0}^n (-1)^k (n-2k)\pi_k(U)$.

In \cite{YM} we present the relation between this operator and projections onto subspaces of fixed ranks. We use this relation to present new class of gauge invariant solutions of Yang-Mills equations.

We can consider operators (\ref{Reyn0}) with the following subsets $S\subseteq \I$:
\begin{eqnarray}
\I_{\Even}=\{A\in\I,\, |A| - \mbox{even}\},\qquad \I_{\Odd}=\{A\in\I,\, |A| - \mbox{odd}\},\nonumber\\
\I_{k}=\{A\in\I,\quad |A|=k\},\qquad k=0, 1, \ldots, n,\nonumber\\
\I_{\overline k}=\{A\in\I,\quad |A|=m\mod 4\},\qquad m=0, 1, 2, 3.\nonumber
\end{eqnarray}
In the last case we use the concept of so-called {\it quaternion type $m$} \cite{QT} of Clifford algebra element, $m=0, 1, 2, 3$. There is a relation between these operators and projective operators onto fixed subspaces of Clifford algebras. This is a subject for further research.

In \cite{funct} we consider operators $\sum_A \gamma^A U \beta_A$ with 2 different sets $\gamma^a$, $\beta^a$, $a=1, \ldots, n$ of Clifford algebra elements that satisfy
$$\gamma^a \gamma^b+ \gamma^b \gamma^a=2\eta^{ab}e,\qquad \beta^a \beta^b+ \beta^b \beta^a=2\eta^{ab}e.$$
We use these operators to prove generalized Pauli's theorem and some other problems about spin groups (see \cite{Pauli1}, \cite{Pauli2}, \cite{Pauli3}, \cite{Pauli4}, \cite{Pauli5}).

The results of this article (especially about the relation between projection operators and averaging operators; solving commutator equations) may be used in computer calculations.

\subsection*{Acknowledgment} The author is grateful to the three referees for their insightful and constructive comments.
Also the author is grateful to N.G.Marchuk for fruitful discussions.


\end{document}